\documentclass[twoside,12pt,reqno]{amsart}
\usepackage[foot]{amsaddr}
\usepackage[unicode=true,pdfusetitle, bookmarks=true,bookmarksnumbered=false,bookmarksopen=false, breaklinks=false,pdfborder={0 0 0},backref=false,colorlinks=false]{hyperref}
\hypersetup{colorlinks,linkcolor=myurlcolor,citecolor=myurlcolor,urlcolor=myurlcolor}
\usepackage{graphics,epstopdf,graphicx,amsthm,amsmath,amssymb,mathptmx,braket,colortbl,color,bm,framed,mathrsfs}
\usepackage[T1]{fontenc}
\usepackage[up]{subfigure}
\usepackage{tikz}
\usepackage{float}
\usepackage{amsfonts}

\definecolor{myurlcolor}{rgb}{0,0,0.9}
\newcommand{\be}{\begin{equation}}
\newcommand{\ee}{\end{equation}}
\newcommand{\beq}{\begin{eqnarray}}
\newcommand{\eeq}{\end{eqnarray}}
\newcommand{\beqs}{\begin{eqnarray*}}
\newcommand{\eeqs}{\end{eqnarray*}}

\newcommand{\tinyspace}{\mspace{1mu}}

\newcommand{\blue}{\textcolor{blue}}

\newcommand{\lrp}[1]{\left(#1\right)}
\newcommand{\lra}[1]{\langle#1\rangle}
\renewcommand{\geq}{\geqslant}
\renewcommand{\leq}{\leqslant}

\newcommand{\abs}[1]{\left\lvert\tinyspace #1 \tinyspace\right\rvert}
\newcommand{\norm}[1]{\left\lVert #1 \right\rVert}

\newcommand{\stlim}{\text{st.-}\lim}
\newcommand{\nn}{\nonumber\\}

\theoremstyle{plain}
\newtheorem{thm}{Theorem}

\newtheorem{prop}[thm]{Proposition}
\newtheorem{cor}[thm]{Corollary}
\newtheorem{remark}[thm]{Remark}

\tikzstyle WL=[line width=10pt,opacity=1.0]
\newcommand{\drawWL}[3]
{
    \draw[white,WL]  (#2) -- (#3);
    \draw[#1] (#2) -- (#3);
}

\newcommand*{\myproofname}{Proof}

\def\ot{\otimes}
\def\complex{\mathbb{C}}

\newcommand{\St}{\mathcal{S}}

\makeatother

\begin{document}

\title{De Finetti Theorems for Braided Parafermions}
\author{Kaifeng Bu$^{+*}$}
\email{kfbu@fas.harvard.edu}
\author{Arthur Jaffe$^{*}$}
\email{jaffe@g.harvard.edu}
\author{Zhengwei Liu$^{-*}$}
\email{zhengweiliu@fas.harvard.edu}
\author{Jinsong Wu$^{\times *}$}
\email{wjs@hit.edu.cn}
\address[*]{Harvard University, Cambridge, MA 02138, USA}
\address[+]{Zhejiang University, China}
\address[-]{Tsinghua University, China}
\address[$\times$]{Harbin Institute of Technology, China}

\begin{abstract}
The classical de Finetti theorem in probability theory relates symmetry under the permutation group with the independence of random variables. This result  has application in quantum information.  Here we study states that are invariant with respect to a natural action of the braid group, and we emphasize the pictorial formulation and interpretation of our results.  
 We prove a new type of de Finetti theorem for the four-string, double-braid group acting on the parafermion algebra to braid qudits, a natural symmetry in the quon language for quantum information.
We prove that a braid-invariant state is extremal if and only if it is a product state. Furthermore, we provide an explicit characterization of braid-invariant states on the parafermion algebra, including finding a distinction that depends on whether  the order of the parafermion algebra is square free. We characterize the extremal nature of product states (an inverse de Finetti theorem).
\end{abstract}

\maketitle

\setcounter{tocdepth}{1}
\tableofcontents
\section{Introduction}
\subsection{Background}
The famous de Finetti theorem in classical probability theory 
clarifies the relationship between permutation symmetry and the independence of a sequence of random variables~\cite{dF31,deFinetti37,Edwin55}. Consequently an infinite sequence of symmetric random variables can be written as a convex combination of an independent identically
distributed (i.i.d.) sequence. 

St{\o}rmer~\cite{Stormer1969}  proposed a non-commutative (quantum) version of the de Finetti theorem, and he demonstrated  that extremal, symmetric states on infinite, tensor-product $C^*$ algebras can be expressed in terms of product states. Other symmetry groups yield non-commutative formulations of de Finetti theorems, and braid invariance has been considered by
Gohm and K\"oster in \cite{Gohm2009, Kostler10}.  The  de Finetti  theorem has been extended to noncommutative probability theory, with a classical probability measure being replaced by  quantum state~\cite{KS09,Curran09,Curran10, GK10, Crismale2012,Banica12,Liu15,Liutran17}.

Diaconis and Freedman established a de Finetti theorem for  a finite (rather than infinite) sequence of exchangeable random variables~\cite{Diaconis1980}. This led to various types of de Finetti theorems in statistical physics and in quantum information~\cite{Hudson1976,Fannes1988,Raggio89,Caves02,Renner05, Renner07, Christandl2007}.  K\"onig and Renner~\cite{Renner05} showed that  any $k$-partite reduced state arising from a state on 
$n$ systems that is  permutation-symmetric, with $k\ll n$, is close to a convex combination of i.i.d. n-partite states.  Here i.i.d. means that the state $\varphi=\rho^{\ot n}$  can be written as a product  of identical copies.

This result is crucial for understanding the structure of permutation-sym{\-}metric states, and especially for the consideration of quantum entanglement of such states~\cite{HorodeckiRMP09}.
The use of such states has application in quantum information processing tasks ranging from entanglement testing~\cite{Brandao2011}, quantum key distribution \cite{Renner2005phd}, quantum hypothesis testing~\cite{Brandao2010}, to quantum  state tomography~\cite{Renner07}, and quantum complexity theory~\cite{Brandao2011, LiPRL2015, Brandao2017}.

Non-abelian statistics of quasiparticle models allow one to perform topological quantum computation, such as  in the zero-mode model for Majorana fermions (the $d=2$ case of parafermions)~\cite{Kitaev03,Nayak08}.
Parafermions, as a generalization of Majoranas, have recently attracted much attention in condensed matter physics \cite{Fendley12,Lindner12,You12} and \cite{Clarke13,Mong14,Klinovaja14,Hutter14}. 
We have given a natural, pictorial representation of the parafermion algebra and showed how this yields a pictorial representation of their Clifford gates~\cite{JLW18,LiuNJP17,Liu17}.

\subsection{New Results}
Here we present a de Finetti theorem for states on  para{\-}fermion algebras of order $d$. In particular, we use the fact that a pair of parafermions of order $d$ generate the 
$d\times d$ matrix algebra $\mathbb{M}_d(\complex)$, that we denote by $PF_2$.  Thus it is natural to consider pairs of parafermions as a unit, and to study double braids that exchange these pairs.  

The infinite parafermion algebra $PF_{\infty}$ is a $\mathbb{Z}_d$-graded, tensor product of algebras $PF_{2}$ of parafermion pairs. 
Here we consider the braid group $\mathbb{B}_{\infty}$, as defined in \S\ref{Sect:Braiding},  acting on pairs of parafermions .  Let $S_{\mathbb{B}_{\infty}}$ denote the states on $PF_{\infty}$ that are invariant under the action of $\mathbb{B}_{\infty}$. In \S \ref{sec:de_f_neu}--\S\ref{sec:de_f_nonneu} we prove: 

\begin{thm}[\textbf{de Finitti  for braided parafermions}]\label{Thm:main1}
Let $\varphi \in S_{\mathbb{B}_{\infty}}$ be a braid-invariant state on  $PF_{\infty}$. 
Then the following are equivalent:
\begin{enumerate}
\item The state $\varphi$ is extremal in the set of states $S_{\mathbb{B}_{\infty}}$ on $PF_{\infty}$.
\item The state $\varphi=\rho^{\otimes\infty}$ is the infinite tensor product of a state $\rho$ on $PF_2$. 
\end{enumerate}
\end{thm}
\blue{We can refine this characterization, depending on the order  $d$ of the algebra $PF_{\infty}$. This results in different restrictions on $\rho$.
As a consequence of 
Theorem \ref{Thm:main1},} any  $\mathbb{B}_{\infty}$-invariant state on $PF_{\infty}$ is in the closure of the convex hull of the product states.
Let $\overline{PF}^{\varphi}_{\infty}$ denote the von Neumann algebra
generated by $PF_{\infty}$ in the Gel{\-}fand-Naimark-Segal (GNS) construction  with respect to the  state $\varphi\in S_{\mathbb{B}_{\infty}}$.  Also let   $(\overline{PF}^{\varphi}_{\infty})^{\mathbb{B}_{\infty}}$ be the fixed point algebra under the action of the braid group $\mathbb{B}_{\infty}$. 
The neutral subalgebra of  $(\overline{PF}^{\varphi}_{\infty})^{\mathbb{B}_{\infty}}$  is the subalgebra generated by monomials in parafermions of degree zero  mod $d$.  

It is interesting that a distinction arises in this characterization,  according to whether or not the order of the parafermion algebra is square free.    (This means that  $d=\prod_{i} p_i$\,, where the primes $p_i$ are distinct.)
Let us now suppose that the degree $d$ of the parafermion algebra is square free.  In this case one finds that extremal, braid-invariant states are neutral and that they  give rise to a factor.  One can refine Theorem \ref{Thm:main1}  as follows:

\begin{thm}\label{Thm:main2}
Let $\varphi\in S_{\mathbb{B}_{\infty}}$ be a braid-invariant state on a parafermion algebra $PF_{\infty}$ of square-free degree $d$.  
The following are equivalent:
\begin{enumerate}
\item  The state $\varphi$ is extremal in $S_{\mathbb{B}_{\infty}}$.
\item  The state $\varphi= \rho^{\otimes\infty}$, where $\rho$ is  a neutral state on $PF_{2}$.
\item The neutral subalgebra of $(\overline{PF}^{\varphi}_{\infty})^{\mathbb{B}_{\infty}}=\complex $.
\item The algebra $(\overline{PF}^{\varphi}_{\infty})^{\mathbb{B}_{\infty}}=\complex $.
\item  The von Neumann algebra $\overline{PF}^{\varphi}_{\infty}$ is a factor.
\end{enumerate}
\end{thm}

In case $d$ is not square free, we give the corresponding characterization braid-invariant states and their von Neumann algebras in Theorems~\ref{thm:main2}--\ref{thm:main3}.  As this requires some additional  terminology that we only introduce later, we postpone these statements to \S\ref{sec:de_f_nonneu}.

A main difference is that an extremal, braid-invariant state is not necessarily neutral, nor is the corresponding algebra necessarily a factor.  
It is  interesting that the de Finetti theorem suggests new methods to prove whether the von-Neumann algebra $\overline{PF}^{\varphi}_{\infty}$ is a factor. 
Finally in Theorem \ref{Thm:Inverse-dF} we characterize the extremal nature of product states (an inverse de Finetti theorem).

\subsection{Organization} 
In \S\ref{sec:intro_para} we define parafermion algebras and their diagrammatic representation which we call parafermion planar para algebras (PAPPA).
In \S \ref{sec:braid_group},  we introduce the four-string braid group $\mathbb{B}_{\infty}$ acting on the parafermion algebra.  We describe the braid using its diagrammatic representation in the PAPPA model.
In \S \ref{sec:lems_tail_alg} we introduce braid-invariant states, the action of shifts,  the tail algebra,   and conditional expectaions onto the tail algebra.  We also introduce the neutral part of the tail algebra, which equals the center of the parafermion algebra.  We derive the independence of the conditional expectation onto the  tail algebra.
In \S \ref{sec:de_f_neu} we prove the de Finetti theorem for parafermion algebras in the case that $d$ is square free. We show that 
 the center of parafermion algebra is equal to the tail algebra of the parafermion alegbra and that  the tail algebra only consists of neutral elements. 
In  \S\ref{sec:de_f_nonneu} we generalize our de Finetti theorem to the case that $d$ is not square free. In this case, the tail algebra does not equal the center of the parafermion algebra, and the tail algebra contains non-neutral elements.
We characterize the tail algebra for extremal, braid-invariant states.

\section{Parafermion algebras and the PAPPA model}\label{sec:intro_para}
The parafermion algebra $PF_m$ of degree $d$ is the $\mathbb{Z}_d$-graded *-algebra generated by $\set{c_j}^m_{j=1}$, with $m$ possibly infinite. We denote the degree of a monomial $A\in PF_m$ by $deg(A)\in \mathbb{Z}_d$.  Since it is often useful to regard the degree as a physical ``charge,'' we sometimes use this term interchangeably  with ``degree.'' 

The generators $c_j$ of the algebra are called parafermions and  satisfy the  canonical parafermion relations (CPRs)
\begin{eqnarray}
\label{ParafermionRelns}
c_jc_k=q\,c_kc_j\;, ~~\text{for}~~ j<k,\quad
c^d_j=I,\quad
c^*_j=c^{-1}_j,
\end{eqnarray} 
where $q=e^{\frac{2\pi i }{d}}$, and $i=\sqrt{-1}$. 
Then 
$PF_{2m}$ is isomorphic to the tensor product of 
$m$-copies of $\mathbb{M}_d(\complex)$ through the ``Jordan-Wigner'' transformation, 
\begin{eqnarray}
PF_{2m}\cong\ot^{m}_{k=1}\mathbb{M}_d(\complex).
\end{eqnarray}
Therefore, the  parafermion algebra $\displaystyle PF_{\infty}:=\lim_{m\to \infty} PF_{m}$ is isomorphic to 
the infinite tensor product of $\mathbb{M}_d(\complex)$,
\begin{eqnarray}
PF_{\infty}\cong {\ot_{k=1}^{\infty}\mathbb{M}_d(\complex)}.
\end{eqnarray}
According to the charge in the parafermion algebra,  it can be decomposed by charge as 
$PF_{\infty}=\oplus^{d-1}_{k=0} PF^k_{\infty}$, where
\begin{eqnarray}
PF^k_{\infty}
=\set{x\in PF_{\infty}| deg(x)=k}.
\end{eqnarray}
The charge-zero subalgebra $PF^0_{\infty}$ is  called the  
neutral subalgebra.

We use a pictorial representation for the element $c^m_j$ introduced in \cite{JLCMP17}, where the algebra is called PAPPA.  We represent  $c_j^m$   by  inserting the label $m$ on the $j^{\rm th}$ string (numbered from left to right).  We place the label of the string position above or below the string, and we omit that label in case this will cause no confusion.  We interpret $m$ as a $\mathbb{Z}_{d}$-valued charge, so we also call PAPPA a ``charged-string model.''  The correspondence between parafermion operators and pictures is:
\[
c_j^m \longleftrightarrow
\raisebox{-.7cm}
{\scalebox{.7}{
\begin{tikzpicture}
\node at (-.75,.5) {$...$};
\node at (.75,.5) {$...$};
 \draw (-1,0)--(-1,1);
 \draw (1,0)--(1,1);
 \draw (-.5,0)--(-.5,1);
 \draw (.5,0)--(.5,1);
 \draw (0,0)--(0,1);
 \node at (0,-.3) {j};
  \node at (-.2,.5) {m};
\end{tikzpicture}}}\;.
\]
Multiplication is designated from top to bottom, corresponding to  algebraic factors written from right to left.
In the PAPPA model, the charged strings 
satisfy the following relations, 
\beq\label{AddCharge}
\textbf{Multiplication:  }&&
\raisebox{-.25cm}
{\scalebox{.7}{
\begin{tikzpicture}
\draw (0,0) --(0,1);
\node (0,0) at (-0.2,0.8) {$m$};
\node (0,0) at (-0.2,0.2) {$n$};
\node (0.45,0.25) at (0.45,0.5) {$=$};
\node (0.9,0.35) at (1.4,0.5) {$m+n$};
\draw (2,0) --(2,1);
\node (1.4,-0.05) at (2.3,0.3) {,};
\end{tikzpicture} }}\qquad
\raisebox{-.25cm}
{\scalebox{.7}{
\tikz{
\node (3.7,0.35) at (4.8,0.5) {$d$};
\draw (5,0) --(5,1);
\node (4.4,0.25) at (5.4,0.5) {$=$};
\draw (5.8,0) --(5.8,1);
\node (4.4,0.25) at (6.2,0.5) {$=$};
\node (3.7,0.35) at (6.5,0.5) {$0$};
\draw (6.7,0) --(6.7,1);
\node (4.95,0) at (7,0.3) {.};
}}}\nonumber\\
\textbf{Para isotopy:  }&&
\raisebox{-.25cm}
{\scalebox{.7}{
\begin{tikzpicture}
\draw (0,0) --(0,01);
\draw (0.2,0) --(0.2,01);
\node at (.5,.5) {$\cdots$};
\node (0,0) at (-0.15,0.2) {$n$};
\draw (0.8,0) --(0.8,1);
\draw (1.15,0) --(1.15,1);
\node (1.2,0) at (1,0.8) {$m$};
\end{tikzpicture}
}}
=  \ \scriptstyle q^{mn}
\raisebox{-.25cm}
{\scalebox{.7}{
\begin{tikzpicture}
\draw (0,0) --(0,01);
\draw (0.2,0) --(0.2,01);
\node at (.5,.5) {$\cdots$};
\node (0,0) at (-0.15,0.8) {$n$};
\draw (0.8,0) --(0.8,1);
\draw (1.15,0) --(1.15,1);
\node (1.2,0) at (1,0.2) {$m$};
\end{tikzpicture}}}\;.
\eeq
\be \label{Equ:twisted tensor product}
\hskip -1.5cm
\textbf{Twisted product:    }\qquad
\raisebox{-.25cm}
{\scalebox{.7}{
\begin{tikzpicture}
\draw (0,0) --(0,01);
\draw (0.2,0) --(0.2,01);
\node at (.5,.5) {$\cdots$};
\node (0,0) at (-0.15,0.5) {$n$};
\draw (0.8,0) --(0.8,1);
\draw (1.15,0) --(1.15,1);
\node (1.2,0) at (1,0.5) {$m$};
\end{tikzpicture}
}}
:=  \ \scriptstyle \zeta^{mn}
\raisebox{-.25cm}
{\scalebox{.7}{
\begin{tikzpicture}
\draw (0,0) --(0,01);
\draw (0.2,0) --(0.2,01);
\node at (.5,.5) {$\cdots$};
\node (0,0) at (-0.15,0.8) {$n$};
\draw (0.8,0) --(0.8,1);
\draw (1.15,0) --(1.15,1);
\node (1.2,0) at (1,0.2) {$m$};
\end{tikzpicture}}}\;.
\ee
Here $\zeta$ is a chosen square root of $q$ such that $\zeta^{d^2}=1$ \cite{JLCMP17,JLW18}.

\section{The four-string braid group $\mathbb{B}_{\infty}$}
\label{sec:braid_group}
We consider the four-string braid group generated by exchanges of pairs of adjacent parafermions; in quantum information this corresponds to braiding  adjacent qudits. 

\subsection{Four string braids}\label{Sect:Braiding}
The algebra generated by  two consecutive generators $(c_j, c_{j+1})$ is isomorphic to $\mathbb{M}_d(\complex)$; in other words $PF_2\cong \mathbb{M}_d(\complex)$. 
Motivated by this picture, we 
consider the action of braid group $\mathbb{B}_{\infty}$ on the pairs of the generators $(c_{2j-1}, c_{2j})$.  In particular, for the parafermion algebra $PF_{2m}$, we introduce the braid group  
\begin{eqnarray}
\mathbb{B}_{2m}=\langle b_1, b_2,..., b_{m-1}\rangle,
\end{eqnarray} 
that is generated by $m-1$ four-string braids. The  braid $b_j$ 
 intertwines the  $(2j-1)^{\rm th}$ and $(2j)^{\rm th}$ strings with the $(2j+1)^{\rm th}$ and $(2j+2)^{\rm th}$  strings. The following picture represents  this (negative) four-string braid action: 

\be
b_{j}
=
\raisebox{-1cm}{\begin{tikzpicture}
\begin{scope}[shift={(4,1)},yscale=-1,xscale=1]
\draw (-2, 0)--(-2, 2); 
\draw (-1, 0)--(-1, 2); 
\draw (4, 0)--(4, 2); 
\draw (5, 0)--(5, 2); 
\node at (-1.5, 1){$\cdots$};
\node at (4.5, 1){$\cdots$};
\node at (0, -.2) {$_{_{_{_{ 2j-1}}}}$};
\node at (1, -.2) {$_{_{_{_{ 2j}}}}$};
\node at (2, -.2) {$_{_{_{_{ 2j+1}}}}$};
\node at (3, -.2) {$_{_{_{_{ 2j+2}}}}$};
\draw (2,0)--(0,2);
\draw (3,0)--(1,2);
\drawWL {}{0,0}{2,2};
\drawWL {}{1,0}{3,2};
\end{scope}
\end{tikzpicture}}\;.
\ee
\medskip

\begin{prop}\label{Prop:FourStringBraidPF}
Each four-string braid $b_{j}$ is an element of $PF_{\infty}$.
 The generators $b_{j}\in\mathbb{B}_{\infty}$ satisfy the double-braid relations:
\begin{eqnarray}\label{DoubleBraidRelations}
b_jb_k=b_kb_j, \quad\text{if } |j-k|>1,\nonumber\\
b_jb_kb_j=b_kb_jb_k,\quad\text{if } |j-k|=1.
\end{eqnarray}

\end{prop}

\begin{proof}
The negative four-string braid $b_{j}$ is the product of four two-string braids, 
\be\label{Double-Braid}
\scalebox{1.2}{
\raisebox{-1cm}{\begin{tikzpicture}
\begin{scope}[shift={(4,1)},yscale=-1,xscale=1]
\draw (2,0)--(0,2);
\draw (3,0)--(1,2);
\drawWL {}{0,0}{2,2};
\drawWL {}{1,0}{3,2};
\end{scope}
\end{tikzpicture}}
\ =\ 
\raisebox{-0.8cm}{\scalebox{1}{\begin{tikzpicture}
\begin{scope}[shift={(8,1)},yscale=-1,xscale=1]
\draw (0.5,0)--(0,0.5);
\drawWL {}{0,0}{0.5,0.5};
\draw (0,0+0.6)--(-0.5,1.1);
\drawWL {}{-0.5,0.6}{0,1.1};
\draw (1,0+0.6)--(0.5,1.1);
\drawWL {}{0.5,0.6}{1,1.1};
\draw (0.5,1.2)--(0,0.5+1.2);
\drawWL {}{0,0+1.2}{0.5,0.5+1.2};
\end{scope}
\end{tikzpicture}}}}\;.
\ee
Each two-string braid can be expressed in terms of the generators of the parafermion algebra, as shown in formula (8.1) of~\cite{JLCMP17}, 
\be
{b_{k}^{(2)}}
=
\raisebox{-1.cm}{\begin{tikzpicture}
\begin{scope}[shift={(4,1)},yscale=-1,xscale=1]
\draw (2,0)--(0,2);
\drawWL {}{0,0}{2,2};
\node at (2.05,2.1) {$_{_{_{_{k+1}}}}$};
\node at (.05,2.05) {$_{_{_{_{_{_{k}}}}}}$};
\end{scope}
\end{tikzpicture}}
\ =\ \frac{\omega^{1/2}}{\sqrt{d}}\sum^{d-1}_{i=0}
\raisebox{-1.2cm}{\scalebox{1}{\begin{tikzpicture}
\begin{scope}[shift={(8,1)},yscale=-1,xscale=1]
\node at (1.8,1.8){$_{_{_{_{ i}}}}$};
\node at (2.8,0.2){$_{_{_{_{ -i}}}}$};
\node at (2.05,2.1) {$_{_{_{_{k}}}}$};
\node at (3.05,2.05) {$_{_{_{_{_{_{k+1}}}}}}$};
\draw (2,0)--(2,2);
\draw (3,0)--(3,2);
\end{scope}
\end{tikzpicture}}}\;.
\ee
Here $\omega=\frac{1}{\sqrt{d}}\sum^{d-1}_{j=0}\zeta^{j^2}$ is shown to be a phase in Proposition 2.15 of \cite{JLCMP17}.  As a consequence the two-string braid $b_{k}^{(2)}$  is a unitary, namely $b_{j}^{(2)\,*}b_{j}^{(2)}=I$, and hence so is the four-string braid $b_{j}$.  Also the left-most single braid $b_{2j-1}^{(2)}$ and the right-most single braid $b_{2j+1}^{(2)}$ 
 illustrated in \eqref{Double-Braid} commute, so their relative vertical order does not matter.  This decomposition shows that $b_{j}$ is in the algebra generated by the four parafermions $c_{2j-1}, c_{2j}, c_{2j+1}, c_{2j+2}$.  

The double-braid relations \eqref{DoubleBraidRelations} are evident from the picture representation of the double braid, and the fact that the single braid satisfies the three Reidermeister moves, see \S8 in~\cite{JLCMP17} and \S3.7 in~\cite{JLW18}.
\end{proof}

\subsection{Action of double braids on parafermions}
The natural action of the double braids $\mathbb{B}_{2m}$ on the parafermion algebra $PF_{2m}$ is the adjoint action $Ad(b_{j})\in Aut(PF_{2m})$. The adjoint action exchanges the pair $(c_{2j-1}, c_{2j})$ with the pair  $(c_{2j+1}, c_{2j+2})$. Thus 
	\be 
	Ad(b_j)(c^{m}_{2j-1}c^{n}_{2j})=b_j(c^{m}_{2j-1}c^{n}_{2j})b^{-1}_j=c^{m}_{2j+1}c^{n}_{2j+2}\;.
	\ee
The diagram that corresponds to this action is: 

\be
\raisebox{-1.5cm}{
\scalebox{.5}{\begin{tikzpicture}

\begin{scope}[shift={(4,-4)},yscale=1,xscale=1]
\node at (0, -.2) {$_{_{_{_{ 2j-1}}}}$};
\node at (1, -.2) {$_{_{_{_{ 2j}}}}$};
\node at (2, -.2) {$_{_{_{_{ 2j+1}}}}$};
\node at (3, -.2) {$_{_{_{_{ 2j+2}}}}$};
\draw (2,0)--(0,2);
\draw (3,0)--(1,2);
\drawWL {}{0,0}{2,2};
\drawWL {}{1,0}{3,2};
\end{scope}
\node at (3.8, -1.9){$_{m}$};
\node at (4.8, -1.1){$_{n}$};
\draw (4,-1)--(4, -2);
\draw (5,-1)--(5, -2);
\draw (6,-1)--(6, -2);
\draw (7,-1)--(7, -2);
\begin{scope}[shift={(4,1)},yscale=-1,xscale=1]
\node at (0, -.2) {$_{_{_{_{ 2j-1}}}}$};
\node at (1, -.2) {$_{_{_{_{ 2j}}}}$};
\node at (2, -.2) {$_{_{_{_{ 2j+1}}}}$};
\node at (3, -.2) {$_{_{_{_{ 2j+2}}}}$};
\draw (2,0)--(0,2);
\draw (3,0)--(1,2);
\drawWL {}{0,0}{2,2};
\drawWL {}{1,0}{3,2};
\end{scope}
\end{tikzpicture}}}
\ =\ 
\raisebox{-1.25cm}{\scalebox{.5}{\begin{tikzpicture}
\begin{scope}
\draw (4,3)--(4, -2);
\draw (5,3)--(5, -2);
\draw (6,3)--(6, -2);
\draw (7,3)--(7, -2);
\node at (5.8, 0.1){$_{m}$};
\node at (6.8, .9){$_{n}$};
\end{scope}
\end{tikzpicture}}}\  \  .
\ee
This equality combines the second Reidermeister move for the braid, as well as the fact that charges pass freely under a braid, see Theorem 8.2 in~\cite{JLCMP17}.  In a similar manner, one can analyze the case with charges on all four strings. The composition of braids translates a sequence of qudits; for example, $Ad(b_{j-1}b_{j})$ acting on the $j-1$ and $j$ qudit spaces, tensored with the identity on the qudit $j+1$, can be pictured as  
\[
\scriptstyle{Ad(b_{j-1}b_{j})}\left(
\raisebox{-2.35cm}{\scalebox{.5}{\begin{tikzpicture}
\begin{scope}
\node at (2, -5.2) {$_{_{_{_{ 2j-2}}}}$};
\node at (3, -5.2) {$_{_{_{_{ 2j-3}}}}$};
\node at (4, -5.2) {$_{_{_{_{ 2j-1}}}}$};
\node at (5, -5.2) {$_{_{_{_{ 2j}}}}$};
\node at (6, -5.2) {$_{_{_{_{ 2j+1}}}}$};
\node at (7, -5.2) {$_{_{_{_{ 2j+2}}}}$};
\draw (2,4)--(2, -5);
\draw (3,4)--(3, -5);
\draw (4,4)--(4, -5);
\draw (5,4)--(5, -5);
\draw (6,4)--(6, -5);
\draw (7,4)--(7, -5);
\node at (1.8, -.6){$_{k}$};
\node at (2.8, -.6){$_{l}$};
\node at (3.8, -.6){$_{m}$};
\node at (4.8, -.6){$_{n}$};
\end{scope}
\end{tikzpicture}}}   \right)\ 
=\
\raisebox{-2.5cm}{
\scalebox{.5}{\begin{tikzpicture}
\begin{scope}[shift={(4,-4)},yscale=1,xscale=1]
\node at (0, -2.2) {$_{_{_{_{ 2j-1}}}}$};
\node at (1, -2.2) {$_{_{_{_{ 2j}}}}$};
\node at (2, -2.2) {$_{_{_{_{ 2j+1}}}}$};
\node at (3, -2.2) {$_{_{_{_{ 2j+2}}}}$};
\node at (-1, -2.2) {$_{_{_{_{ 2j-2}}}}$};
\node at (-2, -2.2) {$_{_{_{_{ 2j-3}}}}$};
\draw (-1, 0)--(1,-2);
\draw (-2, 0)--(0,-2);
\draw (2,0)--(0,2);
\draw (3,0)--(1,2);
\draw (2,0)--(2,-2);
\draw (3,0)--(3,-2);
\drawWL {}{0,0}{-2,-2};
\drawWL {}{1,0}{-1,-2};

\drawWL {}{0,0}{2,2};
\drawWL {}{1,0}{3,2};
\end{scope}
\node at (3.8, -1.4){$_{m}$};
\node at (4.8, -1.4){$_{n}$};
\node at (2.8, -1.4){$_{l}$};
\node at (1.8, -1.4){$_{k}$};

\draw (4,-1)--(4, -2);
\draw (5,-1)--(5, -2);
\draw (6,-1)--(6, -2);
\draw (7,-1)--(7, -2);
\draw (2,1)--(2, -4);
\draw (3,1)--(3, -4);

\begin{scope}[shift={(4,1)},yscale=-1,xscale=1]
\node at (0, -2.2) {$_{_{_{_{ 2j-1}}}}$};
\node at (1, -2.2) {$_{_{_{_{ 2j}}}}$};
\node at (2, -2.2) {$_{_{_{_{ 2j+1}}}}$};
\node at (3, -2.2) {$_{_{_{_{ 2j+2}}}}$};
\node at (-1, -2.2) {$_{_{_{_{ 2j-2}}}}$};
\node at (-2, -2.2) {$_{_{_{_{ 2j-3}}}}$};
\draw (-1, 0)--(1,-2);
\draw (-2, 0)--(0,-2);
\draw (2,0)--(0,2);
\draw (3,0)--(1,2);
\draw (2,0)--(2,-2);
\draw (3,0)--(3,-2);
\drawWL {}{0,0}{-2,-2};
\drawWL {}{1,0}{-1,-2};

\drawWL {}{0,0}{2,2};
\drawWL {}{1,0}{3,2};
\end{scope}
\end{tikzpicture}}}
\ =\ 
\raisebox{-2.45cm}{\scalebox{.5}{\begin{tikzpicture}
\begin{scope}
\node at (2, -5.2) {$_{_{_{_{ 2j-2}}}}$};
\node at (3, -5.2) {$_{_{_{_{ 2j-3}}}}$};
\node at (4, -5.2) {$_{_{_{_{ 2j-1}}}}$};
\node at (5, -5.2) {$_{_{_{_{ 2j}}}}$};
\node at (6, -5.2) {$_{_{_{_{ 2j+1}}}}$};
\node at (7, -5.2) {$_{_{_{_{ 2j+2}}}}$};
\draw (2,4)--(2, -5);
\draw (3,4)--(3, -5);
\draw (4,4)--(4, -5);
\draw (5,4)--(5, -5);
\draw (6,4)--(6, -5);
\draw (7,4)--(7, -5);
\node at (3.8, -.4){$_{k}$};
\node at (4.8, -.4){$_{l}$};
\node at (5.8, -.4){$_{m}$};
\node at (6.8, -.4){$_{n}$};
\end{scope}
\end{tikzpicture}}}\  \  .
\]

\subsection{States and automorphisms}
A state $\varphi$ on $PF_{\infty}$ yields by the GNS construction a Hilbert space $\mathcal{H}$, a $*$-representation $\pi$ of $PF_{\infty}$ on $\mathcal{H}$, and  a cyclic vector $\Omega$ such that $\varphi(x)=\langle\Omega, \pi(x)\Omega\rangle_{\mathcal{H}}$. For simplicity, we denote  $\pi(x)$ acting on $\mathcal{H}$ by $x$. We also use $x$ to denote an element of the von Neumann algebra  $\overline{PF}^{\varphi}_{\infty}$ on $\mathcal{H}$ obtained by closing $PF_{\infty}/\mathcal{N}$ in the sesquilinear form $\lra{x,y}=\varphi(x^{*}y)$ arising from $\varphi$ with null space $\mathcal{N}$.

 If the state $\varphi$ is invariant under a $*$-automorphism $\sigma$ of $PF_{\infty}$,  this determines a $*$-automorphism (that we also denote as $\sigma$) on $\overline{PF}^{\varphi}_{\infty}$ and  an isometry $U$ on $\mathcal{H}$,  that leaves $\Omega$ invariant, and such that $\sigma(x)\Omega=Ux\Omega$.   

\subsection{Braid-invariance and shift invariance}
We are especially concerned with shifts of pairs of generators of the parafermion algebra, as they correspond to the action of the four-string braids.  
Define the double shift $\alpha\in End(PF_{\infty})$ by  
 \be\label{Basic2Shift}
\alpha(c_j)=c_{j+2}\;,
\text{ for all }j\in\mathbb{N}\;.
 \ee
 The picture for the double shift is:
\be
\raisebox{-1cm}{\begin{tikzpicture}
\begin{scope}[shift={(4,1)},yscale=-.8,xscale=.8]
\draw (0, 0)--(0, 2); 
\draw (1, 0)--(1, 2); 
\draw (2, 0)--(2, 2); 
\draw (3, 0)--(3, 2); 
\draw (4, 0)--(4, 2); 
\node at (0, -.2) {$_{_{_{_{ 1}}}}$};
\node at (1, -.2) {$_{_{_{_{ 2}}}}$};
\node at (2, -.2) {$_{_{_{_{ 3}}}}$};
\node at (3, -.2) {$_{_{_{_{ 4}}}}$};
\node at (4, -.2) {$_{_{_{_{ j}}}}$};
\node at (3.5, 1){$\cdots$};
\node at (-0.1, 1) {$_{_{_{_{ i_1}}}}$};
\node at (0.9, 1) {$_{_{_{_{ i_2}}}}$};
\node at (1.9, 1) {$_{_{_{_{ i_3}}}}$};
\node at (2.9, 1) {$_{_{_{_{ i_4}}}}$};
\node at (3.9, 1) {$_{_{_{_{ i_j}}}}$};
\node at (4.75, 1) {$\xrightarrow[]{\alpha}$};
\node at (3.5, 1){$\cdots$};
\node at (-0.1+5.5, 1) {$_{_{_{_{}}}}$};
\node at (0.9+5, 1) {$_{_{_{_{ }}}}$};
\node at (1.9+5.5, 1) {$_{_{_{_{ i_1}}}}$};
\node at (2.9+5.5, 1) {$_{_{_{_{ i_2}}}}$};
\node at (3.9+5.5+0.3, 1) {$_{_{_{_{ i_{j-2}}}}}$};
\draw (0+5+0.5, 0)--(0+5.5, 2); 
\draw (1+5+0.5, 0)--(1+5.5, 2); 
\draw (2+5+0.5, 0)--(2+5.5, 2); 
\draw (3+5+0.5, 0)--(3+5.5, 2); 
\draw (4+5+0.5+0.5, 0)--(4+5+0.5+.5, 2); 
\node at (0+5+0.5, -.2) {$_{_{_{_{ 1}}}}$};
\node at (1+5+0.5, -.2) {$_{_{_{_{ 2}}}}$};
\node at (2+5+0.5, -.2) {$_{_{_{_{ 3}}}}$};
\node at (3+5+0.5, -.2) {$_{_{_{_{ 4}}}}$};
\node at (4+5+0.5+0.5, -.2) {$_{_{_{_{ j}}}}$};
\node at (9, 1){$\cdots$};
\end{scope}
\end{tikzpicture}}\;.
\ee
\medskip

Let $(\overline{PF}^{\varphi}_{\infty})^{\alpha}$  denote the fixed point algebra of $\overline{PF}^{\varphi}_{\infty}$
under the shift
\be
(\overline{PF}^{\varphi}_{\infty})^{\alpha}
:=\set{x\in \overline{PF}^{\varphi}_{\infty}|\alpha(x)=x }\;.
\ee
 A state $\varphi$ on $PF_{\infty}$ will be called $\alpha$-shift-invariant if 
 \be\label{TranslationInvarantState}
\varphi=\varphi\circ \alpha\;.
\ee
Let $S_{\alpha}$ denote the set of $\alpha$-shift-invaraiant states on $PF_{\infty}$.  

Similarly we say that the state $\varphi$ on $PF_{\infty}$ is braid-invariant if it is invariant under the adjoint action of the braid group  $\mathbb{B}_{\infty}$, 
\be\label{BraidInvarantState}
\varphi=\varphi\circ Ad(b)\;,
\ee
for any $b\in \mathbb{B}_{\infty}$. 
Let $S_{\mathbb{B}_{\infty}}$ denote the set of $\mathbb{B}_{\infty}$-invariant states on $PF_{\infty}$.

Let  $(\overline{PF}^{\varphi}_{\infty})^{\mathbb{B}_{\infty}}$ 
denote the fixed point algebra of $\overline{PF}^{\varphi}_{\infty}$
under the adjoint action of the braid group $\mathbb{B}_{\infty}$,
\be
(\overline{PF}^{\varphi}_{\infty})^{\mathbb{B}_{\infty}}
:=\set{x\in \overline{PF}^{\varphi}_{\infty}|Ad(\sigma)(x)=x, \forall \sigma\in \mathbb{B}_{\infty}}\;.
\ee

 \begin{prop}
\label{lem:rep_shift}
For  $x\in\overline{PF}^{\varphi}_{\infty}$, the strong limit of consecutive braidings exists. It equals to the shift of $x$,  
 \begin{eqnarray}\label{AlphaAsBraid}
\alpha(x)=\stlim_{n\to\infty}Ad(b_1b_2...b_n)(x)\in\overline{PF}^{\varphi}_{\infty}\;.
\end{eqnarray}
Any braid-invariant state is also $\alpha$-shift-invariant, $S_{\mathbb{B}_{\infty}}\subset S_{\alpha} $ . 
\end{prop}

\begin{proof}
Any element $x\in PF_{2m}$ satisfies $Ad(b_1b_2...b_m)(x)=\alpha(x)$, and furthermore $Ad(b_1b_2...b_n)(x)=Ad(b_1b_2...b_m)(x)$ for $n>m$.  Thus \eqref{AlphaAsBraid} holds on this dense subalgebra.

If $x_{j}\in PF_{2m_{j}}$ converges strongly to $ x \in \overline{PF}^{\varphi}_{\infty}$, we  claim 
that $\alpha(x_{j})$ converges strongly to a limit that we denote $\alpha(x)$. Let $B_{\ell}=b_{1}\cdots b_{\ell}$ denote the unitary transformation implementing this element of the braid group on the GNS Hilbert space $\mathcal{H}$.  Choose  $y\in PF_{k}$, and without loss of generality, let $k<j<j'$. Then 
\beq\label{CauchyAlphaX}
\norm{(\alpha(x_{j})-\alpha(x_{j'}))y}^{2}_{\mathcal{H}} 
&=&
\norm{(B_{m_{j'}}(x_{j}-x_{j'})B_{m_{j'}}^{*}y}^{2}_{\mathcal{H}} \nn
&=& 
\varphi(y^{*}  B_{m_{j'}}(x_{j}-x_{j'})^{*}(x_{j}-x_{j'})B_{m_{j'}}^{*}y)\nn
&=&
\varphi(\tilde{y}^{*} (x_{j}-x_{j'})^{*}(x_{j}-x_{j'})\tilde{y})\;.
\eeq
Here we use the invariance of $\varphi$ under the braid group and the fact that $\tilde{y}=B_{m_{j'}}^{*}yB_{m_{j'}}$ is independent of $j'$ for $k<j'$.  Thus as a consequence of the strong convergence of $x_{j}$, the difference \eqref{CauchyAlphaX} converges to zero as $j\to\infty$.  As the $y$ range over a dense set of $\overline{PF}^{\varphi}_{\infty}$ this verifies \eqref{AlphaAsBraid}.  It also shows that $\varphi$ is $\alpha$-shift-invariant, for 
\[
\varphi(\alpha(x))=\lim_{j}\varphi(\alpha(x_{j}))=\lim_{j} \varphi(B_{m_{j}}x_{j}B_{m_{j}}^{*})=\lim_{j}  \varphi(x_{j})=\varphi(x)\;.
\]
\end{proof}

\section{The tail algebra for parafermions}\label{sec:lems_tail_alg}
We give the basic definitions of braid-invariance and the tail algebra for parafermions, and we derive some general properties. 
\goodbreak
\subsection{Fundamental Concepts}
\paragraph{\bf Neutral states}
The state $\varphi$  on $PF_{\infty}$ is called neutral, if it vanishes on elements with non-zero charge, namely 
$\varphi(x)=0$ for all $x$ with $deg(x)\neq 0$.

\paragraph{\textbf{The tail algebra}}
Let $\varphi$ be a state on $PF_{\infty}$, and let $H_{\varphi}$ denote the Hilbert space obtained by the GNS construction. Let $\overline{PF}^{\varphi}_{\infty}$ be the von Neumann algebra generated by the representation of $PF_{\infty}$ on $H_{\varphi}$.
Define the tail algebra $PF^{T}$ of the parafermion algebra as  
\begin{eqnarray}
PF^T=\bigcap_n \alpha^n(\overline{PF}^{\varphi}_{\infty}).
\end{eqnarray}

\paragraph{\textbf{Conditional expectation}}
Given an algebra
$\mathcal{A}$ and a subalgebra $\mathcal{B}\subset \mathcal{A}$, a $\mathcal{B}-\mathcal{B}$ bimodule linear map $E:\mathcal{A}\rightarrow\mathcal{B}$ is a conditional expectation if  for all $a\in \mathcal{A}$ and $b,b_1,b_2\in \mathcal{B}$, 
\be\label{ConditionalExpectation}
E(\mathcal{A})=\mathcal{B}\;,\qquad 
E(b)=b\;, \qquad \text{and}\quad
E(b_1ab_2)=b_1E(a)b_2\;.
\ee

\paragraph{\textbf{Charge in the tail algebra}}
To decompose the tail algebra according to the charge,  define charge of an element in the tail algebra to be compatible with the charge of elements in the parafermion algebra.

\subsection{Properties of the tail algebra}
The tail algebra can be characterized using the strong operator topology (SOT). For the shift $\alpha$ defined by \eqref{Basic2Shift} on $PF_{\infty}$,  define the shift-averaging transformation 
\be
\mathfrak{S}_{k}
= \frac{\alpha+\alpha^{2}+...+\alpha^{k}}{k}\;.
\ee
Given a state $\varphi$, we  also have the corresponding  $\alpha$ and $\mathfrak{S}_{k}$ on $\overline{PF}^{\varphi}_{\infty}$.

\begin{prop}\label{prop:rep_ET}
Let  $\varphi$ be an $\alpha$-shift-invariant state on $PF_{\infty}$, and let $x\in \overline{PF}^{\varphi}_{\infty}$. Then $\mathfrak{S}_{k}(x)$ converges strongly to an element in the tail, 
\be\label{TailLimitOperator}
E_T(x)=\stlim_{k\to\infty} \mathfrak{S}_{k}(x)\;.
\ee
The map $E_T$ defines a normal, unital conditional expectation
from $\overline{PF}^{\varphi}_{\infty}$ to $PF^{T}$. Also  
	\be\label{AlphaInside}
	E_T(x)=E_T(\alpha(x))\;,
	\quad\text{and}\quad
	(E_T(x))^{*}=E_T(x^{*})\;.
	\ee 
\end{prop}
\begin{proof}
Clearly $E_T(I)=I$, and  
\be\label{Equ: alpha}
\stlim_{k} \frac{\alpha(x)-\alpha^{(k+1)}(x)}{k}=0\;.
\ee
So if $ E_T(x)= \stlim_{k\to\infty} \mathfrak{S}_{k}(x)$ exists, then $E_T(\alpha(x))=E_T(x)$, and the limit is in the tail algebra $PF^T$.
As $\|\mathfrak{S}_{k}(x)\| \leq \|x\|,$ it is sufficient to show that $ \lim_k \mathfrak{S}_{k}(x) A\Omega$ exists for all monomials $A \in PF_{2m}$, with $m\geq 1$.

We first establish the limit when $x$ is a monomial in $PF_{2m}$, so that for any  $k>m$, and for $r=-deg(x)deg(A)$,
\be
\alpha^k(x) A=q^{r} A \alpha^k(x)\;.
\ee
At most $2m$ terms differ in the sums defining $ \mathfrak{S}_{k}(x)A$ and $q^{r} A  \mathfrak{S}_{k}(x)$, so\footnote{We use $\norm{\ \cdot \ }$ to denote the norm of an operator and $\abs{\,\cdot\, }$ to denote the norm of a vector.}
\be\label{TailLimit-1}
\abs{\mathfrak{S}_{k}(x)A\Omega
 -q^{r}A  \mathfrak{S}_{k}(x)\Omega } 
\leq \frac{2m}{k} \|A\| \  \|x\| \   \abs{\Omega} \to_{k\to\infty} 0\;.
\ee 
%
Since $\varphi$ is $\alpha$-shift-invariant, there is a unitary $U_{\alpha}$ on $H_{\varphi}$ that  implements $\alpha$ and leaves $\Omega$ invariant, namely
$$U_{\alpha} A \Omega=\alpha(A) \Omega.$$
Then
\be\label{TailLimit-2}
\mathfrak{S}_{k}(x)\Omega= \lrp{\frac{U_{\alpha}+U_{\alpha}^{2}+...+U_{\alpha}^{k}}{k}}x\Omega \;.
\ee
Here the unitary $U_{\alpha}$ implements  $\alpha$ and leaves $\Omega$ invariant. 
Now we use the von Neumann mean  ergodic theorem, see  page 407 of \cite{RS55}, to conclude that the vectors \eqref{TailLimit-2} converge strongly as $k\to\infty$.  Hence any $x\in PF_{2m}$ satisfies 
\begin{align}
E_T(x) \Omega &= \lim_k \mathfrak{S}_k(x) \Omega \;,\quad\text{and}\quad
E_T(x) A\Omega = q^r A E_T(x) \Omega \;.
\end{align}
Any operator in $PF_{2m}$ is a finite sum of monomials, so the limit \eqref{TailLimitOperator} exists for operators in $PF_{\infty}$.

Now we show that the limit \eqref{TailLimitOperator} extends to all $ x_{0}  \in \overline{PF}^{\varphi}_{\infty}$. 
 By Kaplansky density theorem, there are $x_j \in PF_{2j}$, $j=1,2,\ldots$, such that $\|x_j \| \leq \|x_0\|$ and $ \stlim_{j} x_j=x_0$. Let  $\ell \in \mathbb{Z}_d$ and $y_{j}= \alpha(x_j) $. Define 
\begin{equation}
y_{j,\ell}=\frac{1}{d}\sum_{k\in \mathbb{Z}_d} q^{k \ell} c_1^{-k} y_j c_1^{k}
\;.
\end{equation}
Then $y_j= \sum_{\ell \in \mathbb{Z}_d} y_{j,\ell} $, with $\text{deg}(y_{j,\ell})=\ell $, and they  satisfy 
\begin{align}
\|y_{j,\ell}\| \leq \|y_{j}\| &\leq \|x_0\| \;, \quad\text{and}\quad
\stlim_{j} y_{j,\ell}=y_{0,\ell} \;.
\end{align}
Whenever $k>m$, 
\be
\alpha^k(y_{j,\ell}) A=q^{- \text{deg}(A) \ell} A \alpha^k(y_{j,\ell})\;.
\ee
Arguing  as  above, we infer that  
\be
\abs{\mathfrak{S}_{k}(y_{j,\ell})A\Omega -q^{-\deg(A) \ell}A \mathfrak{S}_{k}(y_{j,\ell})\Omega}
\leq \frac{2m}{k} \|A\| \  \|x_0\| \  \abs{\Omega}\;.
\ee 
Furthermore $\mathfrak{S}_k(y_{0,\ell})$ has a strong limit $E_T(y_{0,\ell})$, such that
\begin{align*}
E_T(y_{j,\ell}) \Omega &= \lim_k \mathfrak{S}_k(y_{j,\ell}) \Omega \;,\\
E_T(y_{j,\ell}) A\Omega &= q^{- \text{deg}(A) \ell} A E_T(y_{j,\ell}) \Omega \;.
\end{align*}
Therefore $E_T(y_0)$ exists. By \eqref{Equ: alpha}, $E_T(x_0)$ exists and
\be
E_T(x_0)=E_T(y_0)=\sum_{\ell\in \mathbb{Z}_d} E_T(y_{0,\ell}) \;.
\ee

Suppose $\{x_{m}\}$ is a sequence in the unit ball of $\overline{PF}^{\varphi}_{\infty}$, and $ \stlim_{j} x_j=x_0$. To show that $E_{T}$ is normal, it is sufficient to show that $\lim_{j} E_{T}(x_{j})=E_{T}(x_{0})$.
Define $y_{j,\ell}$ as above, so the properties above remain true. Moreover,
\begin{align*}
\abs{\mathfrak{S}_{k}(y_{0,\ell}) \Omega - \mathfrak{S}_{k}(y_{j,\ell})\Omega}
&=\abs{\lrp{\frac{U_{\alpha}+U_{\alpha}^{2}+...+U_{\alpha}^{k}}{k}}(y_{0,\ell}-y_{j,\ell} )\Omega}\\
&\leq \abs{ (y_{0,\ell}-y_{j,\ell} )\Omega} \; .
\end{align*}
So 
\begin{align}
| (E_T(y_{0,\ell})  - E_T(y_{j,\ell}))\Omega | &\leq |(y_{0,\ell}-y_{j,\ell} )\Omega | \;,  \\
|  (E_T(y_{0,\ell})  - E_T(y_{j,\ell})) A \Omega | &=  | A (E_T(y_{0,\ell})  - E_T(y_{j,\ell}))\Omega | \\
&\leq \|A\| \ |(y_{0,\ell}-y_{j,\ell} )\Omega | \;.
\end{align}
Therefore,
\begin{align}
\stlim_{j} E_T(y_{j,\ell})&=E_T(y_{0,\ell}) \;,\quad
\stlim_{j} E_T(y_{j})=E_T(y_{0}) \;, \quad\text{and}\\
\stlim_{j} E_{T}(x_{j}) &=  E_{T}(x_{0})\;.
\end{align}
%
%
%

Finally we verify that the map $x\mapsto E_T(x)$ defines a conditional expectation, by checking the three defining relations in \eqref{ConditionalExpectation}. 
We have shown the first identity in \eqref{AlphaInside}.
Note that  $\mathfrak{S}_{k}(x^{*})=\mathfrak{S}_{k}(x)^{*}$, so taking the limit in $k$ we obtain the second identity in~\eqref{AlphaInside}.
For the third identity, consider  $y_{1},y_{2}\in PF^T$. Then
\[
y_{1}\mathfrak{S}_{k}(x)y_{2} 
= \mathfrak{S}_{k}(y_{1}xy_{2})\;,
\]
using the invariance of $PF^{T}$ under the shift $\alpha$.  As a consequence the $k\to\infty$ limits agree, so 
$x\mapsto E_T(x)$ does define a conditional expectation.  
\end{proof}

\begin{remark}
A combination of Theorem 2.2 in \cite{Gohm2009} and Proposition 7.3 in \cite{Kostler10} also shows that $E_T$ is a conditional expectation onto the tail algebra. 
\end{remark}

\begin{cor} \label{cor:BraidExpTail}
Let $\varphi\in S_{\mathbb{B}_{\infty}}$, then $\varphi=\varphi\circ E_{T}$.
\end{cor}

\begin{proof}
From Proposition \ref{lem:rep_shift} we infer that the state
$\varphi$ is invariant under the action of $\alpha$, and from Proposition~\ref{prop:rep_ET} we infer that  
$\varphi\circ E_T=\varphi$.
\end{proof}

\begin{prop}
The tail algebra $PF^T$ is a commutative $\mathbb{Z}_d$-graded von Neumann algebra, with the charge-$\ell$ part denoted $PF^{T,\ell}$, 
\be
PF^{T}=\bigoplus_{\ell \in \mathbb{Z}_d} PT^{T,\ell}.
\ee
Moreover, $PT^{T,\ell}=0$ when $d \nmid \ell^2$.
\end{prop}

\begin{proof}
As in the proof of Proposition \ref{prop:rep_ET}, for any $x \in PF^{T}$, there are $y_{\ell} \in \overline{PF}^{\varphi}_{\infty}$, $\ell \in \mathbb{Z}_d$, such that
\begin{align}
x=\alpha(x)&=\sum_{\ell \in \mathbb{Z}_d} y_{\ell} \;,\\
c_1 y_{\ell} c_1^{-1}&= q^{\ell} y_{\ell} \;.
\end{align}
\blue{Define $y_{\ell}$ to have charge $\ell$.}
By Proposition \ref{prop:rep_ET},
\begin{align}
x&=\sum_{\ell \in \mathbb{Z}_d} E_T(y_{\ell}) \;,\\
c_1 z_{\ell} c_1^{-1}&= q^{\ell} E_T(y_{\ell}) \;.\text{\color{red}REMOVE}\\
\color{blue} \text{and}\qquad
c_1 x c_1^{-1}&\color{blue}
= \sum_{\ell\in\mathbb{Z}_{d}} q^{\ell} E_T(y_{\ell}) \;.
\end{align}
Therefore, the conjugation by the first parafermion generator $c_{1}$ defines an automorphism
on the tail algebra.   
\blue{So if $x$ has charge $\ell$,} there are operators  $y_{j} \in PF_{2(j+1)}$ with charge $\ell$, for $j\geq 1$, such that 
\be
\stlim_{j\to\infty} y_{j} =x \;.
\ee
Then
\begin{align}
\stlim_{j\to\infty} \alpha^k(y_{j}) =x \;, \\
\stlim_{j\to\infty} \alpha^k(y_{j}^*) =x^* \;.
\end{align}
Note that 
\be
\alpha^k(y_{j}) y_{j}^*=q^{\ell^2} y_{j}^*  \alpha^k(y_{j}) \;, ~\forall~ k > j+1 \;.   
\ee
So $x y_{j}^*=q^{\ell^2} y_{j}^* x$. Then $x \alpha^k(y_{j}^*)=q^{\ell^2} \alpha^k(y_{j}^*) x$, and 
\be
x x^*=q^{\ell^2} x^* x \;.
\ee
Both $x x^*$ and  $x^* x $ are positive operators, so $q^{\ell^2}=1$, namely $d \mid \ell^2$.

If $z \in PF_T$ has charge $\ell'$, then similarly we have that $d \mid (\ell')^2$ and
\be
x z=q^{-\ell \ell'} z x \;.
\ee
Then $d^2 \mid (\ell \ell')^2$. So $d \mid \ell \ell'$, and $xz=zx$.
Therefore $PF^T$ is commutative.
\end{proof}



\begin{prop}\label{Prop:eq_fix_tail}
Given a braid-invariant state $\varphi\in S_{\mathbb{B}_{\infty}}$ and the corresponding tail algebra $PF^{T}$, 
\begin{eqnarray}
PF^T=(\overline{PF}^{\varphi}_{\infty})^{\alpha}=(\overline{PF}^{\varphi}_{\infty})^{\mathbb{B}_{\infty}}\;.
\end{eqnarray}
\end{prop}

This result works in a general situation, see Theorem 0.3 in \cite{Gohm2009}. We give a quick proof for braided parafermions here.

\begin{proof}
We claim that $(\overline{PF}^{\varphi}_{\infty})^{\alpha} \subset PF^T \subset (\overline{PF}^{\varphi}_{\infty})^{\mathbb{B}_{\infty}}\subset (\overline{PF}^{\varphi}_{\infty})^{\alpha}$.  Assume  $x\in (\overline{PF}^{\varphi}_{\infty})^{\alpha}$, then $x=\alpha^n(x) \in \alpha^n(\overline{PF}^{\varphi}_{\infty})$ for any $n$, and in particular $x\in \bigcap_n \alpha^n(\overline{PF}^{\varphi}_{\infty})=PF^T$.
For any $x\in PF^T$, one has  $Ad(b_n)(x)=x$ for any $n$ as 
$x\in \alpha(\overline{PF}^{\varphi}_{\infty})$.
Thus, $x\in (\overline{PF}^{\varphi}_{\infty})^{\mathbb{B}_{\infty}}$, i.e.,
$PF^T\subset  (\overline{PF}^{\varphi}_{\infty})^{\mathbb{B}_{\infty}}$.
Moreover, $(\overline{PF}^{\varphi}_{\infty})^{\mathbb{B}_{\infty}}\subset(\overline{PF}^{\varphi}_{\infty})^{\alpha}$
follows from Proposition \ref{lem:rep_shift}. 
\end{proof}

Given a state $\varphi\in  S_{\mathbb{B}_{\infty}}$, let $PF^{T,0}$ denote the neutral subalgebra of  $PF^T$, and let
$Z(\overline{PF}^{\varphi}_{\infty}) $ denote the center of the von Neumann algebra $\overline{PF}^{\varphi}_{\infty}$.

\begin{prop}\label{Prop:eqv_cen_tail}
Let  $\varphi\in  S_{\mathbb{B}_{\infty}}$. Then
\begin{eqnarray}
PF^{T,0}=Z(\overline{PF}^{\varphi}_{\infty})\subset PF^{T}\;.
\end{eqnarray}
\end{prop}

\begin{proof}
Any neutral $x\in PF^{T}$ commutes with all 
elements in $PF_{\infty}$, and therefore commutes with all elements in  $\overline{PF}^{\varphi}_{\infty}$. Hence $PF^{T,0}\subset Z(\overline{PF}^{\varphi}_{\infty})$.  Furthermore, we infer from  Proposition \ref{Prop:FourStringBraidPF} that any braid $b_{j}\in PF_{\infty}$, and also $b_{j}\in \overline{PF}^{\varphi}_{\infty}$. So if $x\in Z(\overline{PF}^{\varphi}_{\infty})$, one has  $b_{j}x=xb_{j}$, and $x$ is invariant under the adjoint action of every $b_{j}$, and  $x\in (\overline{PF}^{\varphi}_{\infty})^{\mathbb{B}_{\infty}}$.  Thus we infer from Proposition \ref{Prop:eq_fix_tail} that 
$x\in   (\overline{PF}^{\varphi}_{\infty})^{\alpha}
= PF^{T}$, so $PF^{T,0}\subset Z(\overline{PF}^{\varphi}_{\infty})\subset PF^{T}$.

If $x\in Z(\overline{PF}^{\varphi}_{\infty})$, then $c_{1}E_T(x)=E_T(x)c_{1}$.
Since $x\in Z(\overline{PF}^{\varphi}_{\infty})\subset PF^{T}$,  $c_{1}E_T(x) c_{1}^{-1} =q^{\deg(E_T(x))} E_T(x)$. So $\deg(E_T(x)) =0$.  Hence $Z(\overline{PF}^{\varphi}_{\infty})\subset PF^{T,0}$, and $PF^{T,0}= Z(\overline{PF}^{\varphi}_{\infty})$.

\end{proof}

Let $I$ be a subset $I\subset \mathbb{N}$, and let $PF_{I}$ denote the parafermion algebra generated by the $c_{i}$ with $i\in I$.
For two subsets $I, J \subset \mathbb{N}$, let $I<J$ means that 
for $i<j$ for all $i\in I, j\in J$.  Clearly if $x\in PF_{I}$ and $y\in PF_{J}$ with $I<J$, then 
	\[
		xy = q^{\deg(x)\deg(y)}\,yx\;.
	\]

\begin{prop}\label{prop:T-indep}
Let $\varphi\in S_{\mathbb{B}_{\infty}}$,  and let  $x\in PF_{I}$, $y\in PF_{J}$, where $I<J$ or $J<I$ are finite subsets of $\mathbb{N}$.  Then $E_T(xy)=E_T(x)E_T(y)$.  Likewise, if $x_{i}\in PF_{I_{i}}$ for increasing  intervals $I_{i}<I_{j}$ for $i<j$, then
	\be\label{DisjointIntervals}	
	E_{T}(x_{1}\cdots x_{k}) = \prod_{i=1}^{k}  E_{T}(x_{i})\;.
	\ee
\end{prop}
Such independence is called order tail-independence in Ref. \cite{Kostler10}, where a stronger
notion of independence, called full tail-independence,  has also been proposed and investigated. 
It was shown in Theorem 8.1 in \cite{Kostler10} that spreadability implies tail-independence.
Here we give a quick proof for braided parafermions.

\begin{proof}
If $I<J$, there exists $\sigma_n\in \mathbb{B}_{\infty}$
such that $\sigma_n(xy)=x\alpha^n(y)$.  
Thus using from Proposition \ref{Prop:eq_fix_tail} and the invariance of $PF^{T}$ under braids,  
\[
E_T(xy)=E_T(\sigma_{n}(xy))=E_T(x\alpha^n(y))
=E_T(x\mathfrak{S}_{k}(y))=E_T(xE_T(y))\;.
\]
Then using \eqref{AlphaInside}, the $\alpha$-shift-invariance of $PF^{T}$ given in  Proposition~\ref{Prop:eq_fix_tail}, and the fact that $PF^{T}$ is an algebra, we infer for $n,k\in \mathbb{N}$ that 
\beqs
E_T(xy)&=&E_{T}(\alpha^{n}(xE_{T}(y)))
=E_T(\alpha^{n}(x) E_T(y)) = E_T(\mathfrak{S}_{k}(x) E_T(y))\nn
&=&E_T(E_T(x)E_T(y))
=E_T(x)E_T(y).
\eeqs

If $x$ and $y$ have a given degree, then $\deg(E_{T}(x))=\deg(x)$.  So if $J<I$ and $xy=(y^{*}x^{*})^{*}$, we obtain  from the previous case,
\[
E_{T}(xy)=E_{T}((y^{*}x^{*})^{*})=E_{T}(y^{*}x^{*})^{*}= (E_{T}(y)^{*}E_{T}(x)^{*})^{*}=E_{T}(x)E_{T}(y)\;.
\]
The general case for two elements follows by linearity for $x$ and $y$ a sum of components with definite degree.  The case for $k$ ordered elements follows by iteration of the two-element case.
\end{proof}

\begin{prop}\label{prop:Braid-TrivialNeutral}
Let $\varphi $ be an extremal state in $S_{\mathbb{B}_{\infty}}$, then $PF^{T,0}=\mathbb{C}$.
\end{prop}

\begin{proof}
We show that if $PF^{T,0}$ is not trivial, then $\varphi$ is not extremal. 
If $\dim(PF^{T,0}) \neq1$,  there exists a non-trivial projection 
$P\in PF^{T,0}$.  
By Proposition \ref{Prop:eqv_cen_tail}, $PF^{T,0}=Z(\overline{PF}^{\varphi}_{\infty})$. If $\varphi(P)=0$, then
$0\leq \varphi(x^*Px)=\varphi(Px^*x) \leq \sqrt{\varphi(P) \varphi((x^*x)^2)}=0$. So $\varphi(x^*Px)=0$, for any $x \in \overline{PF}^{\varphi}_{\infty}$. Therefore $P=0$ in $\overline{PF}^{\varphi}_{\infty}$, a contradiction. Similarly, if $\varphi(P)=1$, then $P=I$, a contradiction. Therefore $\beta=\varphi(P)\in(0,1)$. 

Proposition \ref{Prop:eq_fix_tail} shows $P$ is invariant under action of 
$\mathbb{B}_{\infty}$. Let $\varphi_1(\cdot)=\frac{1}{\beta}\varphi(P(\cdot))$, and let 
$\varphi_2(\cdot)=\frac{1}{1-\beta}\varphi((1-P)( \cdot))$. Then 
$\varphi_1, \varphi_2\in S_{\mathbb{B}_{\infty}}$ and 
$\varphi=\beta\varphi_1+(1-\beta)\varphi_2$,  which contracts with the fact that 
$\varphi$ is extremal. 
\end{proof}

\section{The inverse de Finetti theorem}
Suppose $\mathcal{A}$ is a finite dimensional matrix algebra and $\rho$ is a state on $\mathcal{A}$. 
Let $\mathcal{A}^{\otimes m}=\ot^{m}_{k=1} \mathcal{A} $ be the $m^{\rm th}$ tensor power of $\mathcal{A}$ and $\mathcal{A}_{\infty}=\ot^{\infty}_{k=1}\mathcal{A}$ be the infinite tensor power of $\mathcal{A}$.
Let $S_{P, \mathcal{A}_{\infty}}$  be the states on the infinite tensor product of $\mathcal{A}$ which are invariant under the permutation group.
The de Finetti theorem said that if $\varphi$ is an extremal point in $S_{P, \mathcal{A}_{\infty}}$, then $\varphi$ is the infinite product state $\prod \rho$, for some state $\rho$ on $\mathcal{A}$. 
It was shown by St{\o}rmer, that any symmetric product state is extremal in $S_{P, \mathcal{A}_{\infty}}$, see Theorem 2.7 in \cite{Stormer1969} for a general result on $C^*$-algebras.


The tensor product of states on $\mathcal{A}$ is called a product state.
The symmetric product state is not an extremal point for the symmetric states on finite tensor products. 
We have the following extremal condition for finite symmetric product states, which can be considered as an inverse de Finetti theorem on finite tensors: 

\begin{thm}\label{Thm: finite inverse DF}
Suppose $\mathcal{A}$ is a finite dimensional matrix algebra.
Let $\St$ be the space of states on $\mathcal{A}$.
For $m\geq 2$, take $\mathcal{A}^{\otimes m}=\ot^{m}_{k=1} \mathcal{A} $.
For a state $\phi \in \St$, if $\mu$ is a probability measure of $\St$, such that,
$$\phi^m= \int_{\rho \in \St } \rho^m d\mu(\rho)$$
on $\mathcal{A}^{\otimes m}$, then $\mu$ is the Dirac measure at $\phi$.
\end{thm}

\begin{proof}
Take the restriction on $\mathcal{A}^{\otimes 2}$, we have that 
$$\phi^2= \int_{\rho \in \St } \rho^2 d\mu(\rho).$$
Let $D_{\rho}$ be the density matrix of $\rho$. Then
$$D_\phi \otimes D_{\phi}= \int_{\rho \in \St } D_\rho \otimes D_{\rho} d\mu(\rho).$$
Let $P$ be the range projection of $D_{\phi}$, and $Q=I-P$, then 
$$\int_{\rho \in \St } QD_\rho Q \otimes QD_{\rho}Q d\mu(\rho)=0.$$
Therefore, $QD_\rho Q=0$ for almost all $D$. So $QD=DQ=0$, and $D=PDP$.
Without loss of generality, we assume that $P=I$, then $D_\phi$ is invertible.
Take $C_{\rho}=D_\phi^{-1/2} D_{\rho} D_\phi^{-1/2}$.
Then 
$$I \otimes I= \int_{\rho \in \St } C_\rho \otimes C_{\rho} d\mu(\rho).$$
Let $tr$ be the tracial state on $\mathcal{A}$. Then taking the trace of the tensor/product in the above formula,
\begin{align*}
 \int_{\rho \in \St } tr(C_{\rho})^2 d\mu(\rho)&=tr(I)^2=1 \;,\\
 \int_{\rho \in \St } tr(C_{\rho}^2) d\mu(\rho)&=tr(I^2)=1 \;.
\end{align*}
So
$$ \int_{\rho \in \St }( tr(C_{\rho}^2)-tr(C_{\rho})^2) d\mu(\rho)=0 $$
Note that the covariance of $C_{\rho}$ is 
$$Cov(C_{\rho})=tr(C_{\rho}^2)-tr(C_{\rho})^2=tr((C_{\rho}-tr(C_{\rho}))^2 ) \geq 0.$$
So $Cov(C_{\rho})=0$, and $C_{\rho}=tr(C_{\rho})$ for almost all $\rho$.
Then $D_\rho=tr(C_{\rho}) D_{\phi}$.
Both $D_\rho$ and $D_\phi$ have trace one, so $D_\rho=D_\phi$.
Therefore $\mu$ is a Dirac measure at $\rho$.
\end{proof}

The above proof also applies to infinite tensors and we recover the result of 
St{\o}rmer for matrix algebras:

\begin{thm}\label{Thm:Inverse-dF}
Suppose $\rho$ is a state on a finite dimensional matrix algebra $\mathcal{A}$.
Then the product state $\prod \rho$ is extremal in $S_{P,\mathcal{A}_{\infty}}$.
\end{thm}

\begin{proof}
By the de Finetti theorem, any extremal point in $S_{P,\mathcal{A}_{\infty}}$ is a product state $\rho^{\infty}$, for some $\rho \in \St$.
Note that $\rho_j \to \rho $ weakly on $\mathcal{A}$ iff  $\rho_j^{\infty} \to \rho^{\infty} $ weakly on $\mathcal{A}_{\infty}$. Thus the space of infinite symmetric product states has the same weak topology as $\St$.
By the Choquet-Bishop-de Leeuw theorem, any state in $S_{P,\mathcal{A}_{\infty}}$ is 
$$\int_{\rho \in \St } \rho^\infty d\mu(\rho),$$
for some probability measure $\mu$ on $\St$.
Therefore, any convex combination is also of the above form.
For a state $\phi \in \St$, if
$$\phi^{\infty}=\int_{\rho \in \St } \rho^\infty d\mu(\rho),$$
then 
$$\phi^{2}=\int_{\rho \in \St } \rho^2 d\mu(\rho).$$
By Theorem \ref{Thm: finite inverse DF}, $\int_{\rho=\phi} d\mu(\rho)=1$.
Therefore, $\phi^{\infty}$ is extremal in $S_{P,\mathcal{A}_{\infty}}$..
\end{proof}

\color{black}

\section{The de Finetti theorem: square-free $d$, the neutral case}
\label{sec:de_f_neu}
As explained in the introduction, there are two possible outcomes according to whether  the degree $d$ of the parafermion algebra is square free.  
 In this section we investigate the square-free case. We first prove that the center of the representation of the parafermion algebra equals the tail algebra, and that the tail algebra is neutral.  
 
 Several  equivalent characterization  of the extremal state in $S_{\mathbb{B}_{\infty}}$
follow from these results, and we precisely characterize the corresponding  braid-invariant states. We call Theorem \ref{thm:main1} the de Finetti theorem for the para{\-}fermion algebra with square-free degree $d$.

\begin{thm}\label{thm:neutral}
Let  $d=\prod p_i$ be square free, and let $\varphi\in  S_{\mathbb{B}_{\infty}}$  be braid invariant. Then 
$\varphi$ is neutral:  $\varphi(x)=0$, whenever  $deg(x)\neq 0$.  
\end{thm}

\begin{proof}
Denote
$C^{mn}_i=c^{m}_{2i-1}c^{n}_{2i}$ .
We first  prove that for $j<k$,  
\be\label{ElementaryVanishing}
\varphi(E_T(C^{mn}_j)E_T(C^{mn}_k)^*)=0\;,
\quad\text{unless } m+n= 0~(\text{mod }d)\;.
\ee
Using Corollary \ref{cor:BraidExpTail}, with  Proposition~\ref{prop:T-indep}, and the second identity in  \eqref{AlphaInside}, we have 
\beq\label{Expectation-1}
\varphi(C^{mn}_j(C^{mn}_k)^*)
&=&\varphi(E_T(C^{mn}_j(C^{mn}_k)^*))
=\varphi(E_T(C^{mn}_j)E_T(C^{mn}_k)^*)\nn
&=&\varphi(E_T(C^{mn}_k)E_T(C^{mn}_k)^*)\geqslant0.
\eeq
In the last line we use $\varphi\geqslant0$, as well as $E_T(C^{mn}_{j})=E_T(C^{mn}_{k})$. This property follows from the first identity in \eqref{AlphaInside}, for   
\[
E_T(C^{mn}_{j})
=\alpha^{k-j}(E_T(C^{mn}_{j}))
=E_T(\alpha^{k-j}(C^{mn}_{j}))
=E_T(C^{mn}_{k})\;.
\]

On the other hand, the parafermion relations \eqref{ParafermionRelns}, and reasoning  similar to the proof of \eqref{Expectation-1},  show that 
\beq\label{Expectation-2}
\varphi(C^{mn}_j(C^{mn}_k)^*)
&=&q^{-(m+n)^2}\varphi((C^{mn}_k)^*C^{mn}_j)\nn
&=&q^{-(m+n)^2}\varphi(E_T((C^{mn}_k)^*C^{mn}_j))\nn
&=&q^{-(m+n)^2}\varphi(E_T((C^{mn}_k)^*C^{mn}_j))\nn
&=&q^{-(m+n)^2}\varphi(E_T(C^{mn}_k)^*E_T(C^{mn}_j))\nn
&=&q^{-(m+n)^2}\varphi(E_T(C^{mn}_k)^*E_T(C^{mn}_k))\;,
\eeq
where $\varphi(E_T(C^{mn}_k)^*E_T(C^{mn}_k))\geqslant 0$.  
Comparing \eqref{Expectation-1} with \eqref{Expectation-2},  we infer that either $q^{(m+n)^2}>0$, or else $\varphi(C^{mn}_j(C^{mn}_k)^*)=0$.  

In our case  $0\leqslant m, n \leqslant d-1$.  As we assume that $d$ is square free, $(m+n)^{2}=0 \text { (mod }d)$ is equivalent to $(m+n)=0 \text{ (mod }d)$. Thus \eqref{ElementaryVanishing} holds. 

If $m+n\neq0$ mod $d$, we have shown in addition that 
\be\label{CConservation}
\varphi(E_T(C^{mn}_k)E_T(C^{mn}_k)^*)=0 \;.
\ee
So $E_T(C^{mn}_k)=0$, and 
\be
\varphi(C_{j}^{mn})=\varphi(C_{k}^{mn})=\varphi(E_{T}(C_{k}^{mn}))=0.
\ee

Any element $x\in PF_{\infty}$ can be expressed as the linear combination of the products 
 $\prod_{k}C^{m_kn_k}_k$, where we take the product in the order of increasing $k$ from left to right.  By \eqref{DisjointIntervals} we have
 \be
 \varphi(\prod_{k}C^{m_kn_k}_k)
=  E_{T}(\prod_{k} C^{m_kn_k}_k)
 = \prod_{k} E_{T}(C^{m_kn_k}_{k})
 =0\;,
 \ee
 unless each $(m_{k}+n_{k})=0 \text{ (mod } d)$, for all $k$.
\end{proof}

\begin{cor}\label{cor:eqv_cen_tail}
Let $d=\prod p_i$ be square free, and let  $\varphi\in  S_{\mathbb{B}_{\infty}}$. Then 
\begin{eqnarray}
PF^{T,0}=Z(\overline{PF}^{\varphi}_{\infty})=PF^T
=(\overline{PF}^{\varphi}_{\infty})^{\alpha}=(\overline{PF}^{\varphi}_{\infty})^{\mathbb{B}_{\infty}}.
\end{eqnarray}
\end{cor}
\begin{proof}
With the results of Propositions \ref{Prop:eq_fix_tail}--\ref{Prop:eqv_cen_tail}, we only need to show that $PF^{T,0}=PF^{T}$.  For this use the fact that we have shown in Theorem~\ref{thm:neutral} that the generators of the tail algebra are all neutral.
\end{proof}

\begin{thm}[\bf Neutral Case]\label{thm:main1}
Let $\varphi\in S_{\mathbb{B}_{\infty}}$ be a braid-invariant state on the parafermion algebra $PF_{\infty}$ of square-free order $d$. Then  (1)--(5) are equivalent conditions:

\begin{enumerate}
\item[(1)] The state $\varphi$ is extremal in $S_{\mathbb{B}_{\infty}}$.
\item[(2)]  The state $\varphi=\prod_j \rho$ is a product, where $\rho$ is a neutral state on $PF_2$.
\item[(3)]  The neutral part of the tail algebra $PF^{T,0}=\complex $.
\item[(4)] The tail algebra $PF^T=\complex $.
\item[(5)] The von Neumann algebra $\overline{PF}^{\varphi}_{\infty}$ is a factor.
\end{enumerate}

\end{thm}
\begin{proof}
$(1)\Rightarrow (3)$:
This is established in Proposition \ref{prop:Braid-TrivialNeutral}.

$(4)\Rightarrow (2)$:
Since any element $x\in PF_{\infty}$ can be expressed as the linear combination 
of the form $\displaystyle \prod_k C^{m_kn_k}_k$, we only need to prove that 
\begin{eqnarray}\label{eq:state_indep}
\varphi(\prod_k C^{m_kn_k}_k)
=\prod_k  \varphi (C^{m_kn_k}_k).
\end{eqnarray}
But this is a consequence of Corollary \ref{cor:BraidExpTail}, which ensures that $\varphi=\varphi\circ E_{T}$, and Proposition \ref{prop:T-indep}, which shows that $\varphi$ factors on non-overlapping elements of the tail algebra.

$(2)\Rightarrow(1)$: 
Suppose $\varphi=\prod \rho=\lambda_1 \rho_1 +\lambda_2 \rho_2$, 
where $\lambda_1+\lambda_2=1$, $\lambda_1,\lambda_2>0$, and $\rho_1,\rho_2 \in S_{\mathbb{B}_{\infty}}$.
Let 
$A=D^{\otimes m}$. Similar to the proof in Theorem \ref{Thm:Inverse-dF}, we have 
 $\rho_1=\rho_A\otimes \rho_A=\rho^{\otimes 2m}$ on $PF_{4m}$.
Let $m\to \infty$, we have  $\rho_1=\prod \rho$.
Therefore $\prod \rho$ is extremal.

Finally the equivalence of (3), (4), and (5) follows from 
Corollary \ref{cor:eqv_cen_tail}.
\end{proof}

Since every state  in $S_{\mathbb{B}_{\infty}}$ 
can be written as a limit of  the convex combination of the extremal states, 
such state is in the closure of the convex hull of the product states. This is the de Finetti theorem on parafermion algebra with square-free degree.
\begin{thm}
Let $d$ be square free.  Then 
each state $\varphi\in S_{\mathbb{B}_{\infty}}$ on the $\mathbb{Z}_{d}$ graded parafermion algebra is neutral,
and it can be expressed as the limit of  convex combinations of product states $\prod \rho$, where $\rho\in S(PF_2)$
is neutral.
\end{thm}

\begin{cor}
If $\varphi'=\frac{1}{2}(\prod\rho+\prod\tau)$ where 
$\rho, \tau\in S(PF_2)$ are distinct neutral states, then
the von Neumann algebra  $\overline{PF}^{\varphi'}_{\infty}$
is not a factor.
\end{cor}

\section{The de Finetti theorem: the non-neutral case}\label{sec:de_f_nonneu}
In this section, we consider the case when the square of some prime $p$  divides $d$, so  $d=p^{2}d_{1}$. Let $p_{0}$ denote the smallest natural number such that $d | p_0^2$. Then $p_0 < d$ and $p_0 | d$.  All the preliminary results in \S\ref{sec:lems_tail_alg} hold in this case.  The difference here is that when $d$ contains a square,  the tail algebra may not be neutral.
First, we consider 
the special case 
when the neutral part of the tail algebra is trivial.

\begin{prop}\label{lem:tail_noneu}
Let $\varphi\in S_{\mathbb{B}_{\infty}}$ and let $PF^{T,0}=\complex $, then there exists some $m_0\in \mathbb{N}$, with $p_0 | m_0$, $m_0 | d$, and such that
\begin{eqnarray}\label{eq:tail_non}
PF^T=\bigoplus_{j=1}^{\frac{d}{m_0}} PF^{T,jm_0},
\quad \text{and where }\dim( PF^{T, jm_{0}})=1\;.
\end{eqnarray}
\end{prop}

\begin{proof}
We first show that $PF^{T}$ has a component with a charge $m>0$, only when $p_{0}|m$.  Suppose that $p_{0}$ does not divide $m$, $x\in PF_{I}$, and $E_{T}(x)\in PF^{T,m}$. Then for some $k$ and some $J>I$, one has $\alpha^{k}(x)\in PF_{J}$.  Using Proposition \ref{prop:T-indep} and Proposition \ref{prop:rep_ET}, we infer
 $E_T(x^*\alpha^k(x))=E_T(x)^*E_T(x)\geqslant0$ and 
$E_T(\alpha^k(x)x^*)=E_T(x)E_T(x)^*\geqslant0$. Besides, from the canonical parafermion relations,
$x^*\alpha^k(x)=q^{-m^2}\alpha^{k} (x)x^*$. Thus  either 
 $q^{-m^2}>0$ or else $E_{T}(x)=0$. The only positive value on the unit circle in the complex plane is $1$.  But the condition $p_{0}\nmid m$ means $q^{-m^{2}}\neq1$. So  $E_{T}(x)=0$, and $\dim(PF^{T,m})=0$.

On the other hand suppose that  $p_0|m$.  We show that $\dim( PF^{T,m})\leqslant 1$.   In fact, for $A\in PF^{T,m}$, the quadratic expressions $A^*A$ and $AA^{*}$ are both neutral.  Since we have assumed that $PF^{T,0}=\complex$, then 
$A^*A=\lambda \mathbb{I} $,  and also $AA^{*}=\lambda'I$, with $\lambda,\lambda'\geqslant 0$, and equal to zero only if $A=0$.  But then $(A^{*}A)^{2}= \lambda^{2}\mathbb{I}=\lambda\lambda' \mathbb{I}$, so $\lambda=\lambda'$.  Therefore either $A=0$ or $U_{A}=A/\norm{A}$ is unitary.

Thus for any two non-zero elements $A,B\in PF^{T,m}$, the neutral unitary $U_{A}^{*}U^{\phantom{*}}_{B}=e^{i\theta}$ is a phase, and $U_{B} =e^{i\theta} U_{A}$. In particular $\dim( PF^{T,m})\leqslant 1$ as claimed. Thus
 \begin{eqnarray*}
PF^T=\bigoplus_{k=1}^{\frac{d}{p_0}} PF^{T,kp_{0}},
\quad\text{where }\dim(PF^{T,kp_{0}})  \leqslant 1\;.
\end{eqnarray*}
Since the dimension of $PF^{T,kp_{0}}$ is either 0 or 1, all $kp_0$ such that 
$dim (PF^{T,kp_{0}})=1$ form a subgroup of $\mathbb{Z}_d$. There exists a smallest such number
$k_0$ that divides all k with $dim (PF^{T,kp_{0}})=1$.
Hence $PF^T$ can be written as
\begin{eqnarray*}
PF^T=\bigoplus_{j=1}^{\frac{d}{m_0}} PF^{T,jm_0}\;,
\quad\text{with }\dim{PF^{T,jm_0}}=1\;,
\end{eqnarray*} 
where $m_0=k_0p_0$.
\end{proof}

Define $PF_{2}^{p_0\mathbb{Z}}=\{x \in PF_2 ~|~    \text{deg}(x) \in p_0\mathbb{Z}/ d\mathbb{Z}\}.$
Then for any $x,y \in PF_{2}^{p_0\mathbb{Z}}$, 
$$\alpha^j(x) \alpha^k(y)=q^{\text{deg}(x)\text{deg}(y)} \alpha^k(y) \alpha^j(x)= \alpha^k(y) \alpha^j(x), \forall ~ j\neq k.$$
Therefore for any state $\rho$ on $PF_2$ with density matrix $D$ in $PF_{2}^{p_0\mathbb{Z}}$,
the product state $\prod \rho$ is a well-defined state on ${PF}_{\infty}$.

\begin{thm}\label{thm:main2}
Let $\varphi\in S_{\mathbb{B}_{\infty}}$.  Then the following statement 
are equivalent:
\begin{enumerate}
\item[(1$'$)]  $\varphi$ is extremal in $S_{\mathbb{B}_{\infty}}$.

\item[(2$'$)]  $\varphi=\prod \rho$, where $\rho$ is a state on $PF_{2}$ with a density matrix in $PF_{2}^{p_0\mathbb{Z}}$.
\end{enumerate}
\end{thm}

\begin{proof}
(2$'$) $\Rightarrow$ (1$'$): 
The proof is the same as the proof of $(2)\Rightarrow(1)$
in Theorem~\ref{thm:main1}.

(1$'$)$\Rightarrow$ (2$'$): 
If $\varphi$ is extremal, then Proposition \ref{prop:Braid-TrivialNeutral} shows that  $PF^{T,0}=\complex $. 
Using Proposition \ref{lem:tail_noneu}, we infer \eqref{eq:tail_non}, and that 
$PF^{T,jm_{0}}$ is generated by a unitary $U_{j}$ with charge $j$.  We can choose $\{U_{j}\}$, so that the set $G=\{U_{j}\}$ is a cyclic group. The restriction $\varphi_{G}$ of the state $\varphi$ to $G$ is a convex linear combination of characters $\chi_{j}$ on $G$, namely $\varphi_G=\sum_j\lambda_j \chi_j$, where $\lambda_{j}\geqslant0$ and $\sum_{j}\lambda_{j}=1$.

Each character $\chi_i$ can be extended to a state on $PF^T$, the group algebra of $G$. It can also be extended to a state on the parafermion algebra,
\be
\chi_i(x):=\chi_i(E_T(x)), ~\forall x \in \overline{PF}^{\varphi}_{\infty}.
\ee
as $E_T$ is a conditional expectation.
Then $\chi_i(x)=0$, for any element $x\in \overline{PF}^{\varphi}_{\infty}$ orthogonal to $PF^T$. By Proposition \ref{Prop:eq_fix_tail}, $PF^{T}$ is $\mathbb{B}_\infty$ invariant, so $\chi_i$ is $\mathbb{B}_\infty$ invariant.
Since $\varphi$ is extremal,
 $\varphi_G$ is equal to some character $\chi=\chi_i$, otherwise the state $\varphi$
can be written as the convex combination of the $\mathbb{B}_{\infty}$-invariant states.
Since $\chi$ is a character, thus 
\be
\chi(E_T(x)E_T(y))=\chi(E_T(x))\chi(E_T(y)).
\ee
That is, $\varphi(E_T(x)E_T(y))=\varphi(E_T(x))\varphi(E_T(y)).$
By Proposition \ref{prop:T-indep},  $\varphi=\prod \rho$ for some state $\rho$ on $PF_{2}$.

Let $D=\sum_{k \in \mathbb{Z}_d} D_{k}$ be the density matrix of $\rho$, and $D_{k}\neq0$ in $PF_2^k$.
Then $\varphi(D_k^*\alpha(D_k^*))=\rho(D_k^*)^2>0.$
On the other hand 
$$\varphi(D_k^*\alpha(D_k^*))=q^{k^2} \varphi(\alpha(D_k^*)D_k^*)=q^{k^2}\rho(D_k^*)^2.$$
So $q^{k^2}=1$ and $p_0 | k$. Therefore $D\in PF_2^{p_0\mathbb{Z}}$.
\end{proof}

\begin{thm}\label{thm:main3}
Given a state $\varphi\in S_{\mathbb{B}_{\infty}}$, then the following statements  are equivalent:
\begin{enumerate}
\item[(3$'$)]  $PF^{T,0}=\complex $.
\item[(4$'$)] $\overline{PF}^{\varphi}_{\infty}$ is a factor.
\item[(5$'$)] $\varphi=\sum_{j}\lambda_{j}\chi_i\circ \prod \rho$, namely for any homogenous $x \in \overline{PF}^{\varphi}$,
\be 
\varphi(x)=\sum_{j}\lambda_{j}\chi_i(\text{deg}(x)) (\prod\rho)(x),
\ee
and the density matrix of $\rho$ is in $PF_2^{p_0\mathbb{Z}}$.
\end{enumerate}
\end{thm}
\begin{proof}
Proposition \ref{Prop:eqv_cen_tail} ensures the equivalence between (3$'$)  and (4$'$).

(3$'$)$\Rightarrow$ (5$'$):  If $PF^{T,0}=\complex 1$, then based on the
proof of (1$'$)$\Rightarrow$ (2$'$),  the state restricted in the 
cyclic group G $\varphi_G$ can be expressed as the convex combination 
of characters $\sum_i\lambda_i \chi_i$.
For any $A\in P_I$  and $B\in P_J$, $I<J$. 
We only need to consider the case where $m_0|\text{deg}(A), m_0|\text{deg}(B)$.
Thus
\begin{eqnarray*}
\varphi(AB)&=&\varphi(E_T(AB))\\
&=&\sum_{i}\lambda_i\chi_i(E_T(AB))\\
&=&\sum_i\lambda_i\chi_i(\text{deg}(A)\text{deg}(B))\chi_0(E_T(AB)).
\end{eqnarray*}
Due to the proof of (1$'$)$\Rightarrow$ (2$'$), each irreducible character
corresponds to one state  which can be
written as the product state. Thus, the state $\chi_0\circ E_T$
can be written as $\prod\rho$, and the density matrix of $\rho$ is in $PF_{2}^{p_0\mathbb{Z}}$. Therefore 
$\varphi=\sum_i\lambda_i\chi_i\circ \prod\rho$.

(5$'$) $\Rightarrow$ (3$'$): 
Let $\{b_k\}_{k\in \mathbb{Z}_d}$ be an orthonormal basis of $PF_2$ with inner product $\langle x,y \rangle=\varphi(y^*x)$, such that $b_0=I$.
Since $\varphi$ can be written as 
$\sum_i\lambda_i \chi_i\circ \prod\rho$, then 
any element $x\in PF^{T,0}$ has the orthogonal decomposition
\begin{eqnarray}
x=\beta_0 I + \sum_{i=1}^{\infty} x_i,
\end{eqnarray}
where
\be
x_i=\sum_{k=1}^{d-1} \beta_{i,k} \alpha^{i-1}(b_k) \alpha^{i}(x_{i,k}),
\ee
for some $\beta_{i,k}\in \mathbb{C}$, $x_{i,k} \in \overline{PF}^{\varphi}$, such that $\text{deg}(x_i)=0$ and $\varphi(x_j^*x_i)= \prod\rho(x_j^*x_i)=0$, for $i\neq j$. 
Besides, $\alpha(x)=x$,
which implies that  $x_i=0$ for $i \geqslant 1$. 
That is, all the element in $PF^{T,0}$ is in proportion to 
the identity. Thus, we obtain (5$'$).
\end{proof}

Besides, the condition (2$'$) can always implies the condition 
(5$'$), these five conditions can be summarized as follows.
\begin{eqnarray*}
(1')&\Leftrightarrow& (2')\\
&\Downarrow&\\
(3')\Leftrightarrow &(4')&\Leftrightarrow (5')
\end{eqnarray*}
That is, for the extremal $\mathbb{B}_{\infty}$-invariant state, the corresponding 
tail algebra $PF^T$  can be decomposed as \eqref{eq:tail_non}.

\section{Summary}\label{sec:con}
We have proposed and proved a new type of de Finetti theorem for the parafermion algebra $PF_{\infty}$ with respect to the 
action of braid group $\mathbb{B}_{\infty}$ that braids qudits (pairs of parafermions). We have two results based on whether or not the  degree $d$ of the given parafermion algebra is square free.   In both cases, we characterize  the extremal, braid-invariant states, and show that these properties are equivalent to the states being product states on $\prod PF_{2}\cong\prod\mathbb{M}_{d}$. 

In the square-free case, we have found other equivalent conditions for a braid-invariant state $\varphi$ to be extremal. It  is surprising that one such condition is that  the von Neumann algebra $\overline{PF}^{\varphi}_{\infty}$ is a factor; so if $\varphi$ is the convex combination of extremal, braid-invariant states, then 
 the corresponding von Neumann algebra of parafermions  is not a factor.
%
%
Since the parafermion algebra has been used to establish a framework  of topological quantum information
theory, the de Finetti theorem  in this work can  shed insight on the topological quantum information
theory. 
%
%
%
%
\section*{Acknowledgment}
We are grateful to Claus K\"ostler, as well as to the anonymous referee, for a number of helpful comments on earlier versions of our paper.
We thank the Templeton Religion Trust for  support of this research under grants TRT0080 and TRT0159.  Kaifeng Bu was partially supported with an Academic Award for Outstanding Doctoral Candidates from Zhejiang University.
Arthur Jaffe was partially supported by ARO Grant W911NF1910302.  Jinsong Wu was supported by NSFC 11771413.

\end{document}